%% file: TolerantJuntaLowerBound.tex
\documentclass[11pt]{article}
\usepackage[utf8]{inputenc}
\usepackage[left=1in, right=1in, top=1in, bottom=1in]{geometry}
\usepackage{amsmath, amssymb, amsthm}
\usepackage{hyperref}
\usepackage{enumitem}
\usepackage{mathtools}

\usepackage[numbers]{natbib}

\DeclareMathOperator{\dist}{dist}

\newcommand{\E}{\mathbb{E}}

\newcommand{\eps}{\varepsilon}

\DeclareMathOperator{\bern}{Bernoulli}

\DeclareMathOperator{\maj}{MAJ}

\newcommand{\poly}{\textsf{poly}}
\newcommand{\ALG}{\textsf{ALG}}

\newcommand{\calD}{\mathcal{D}}

\newcommand{\calJ}{\mathcal{J}}

\newcommand{\bx}{\mathbf{x}}
\newcommand{\bff}{\mathbf{f}}
\newcommand{\bgg}{\mathbf{g}}

\newcommand{\bi}{\mathbf{i}}

\newcommand{\bJ}{\mathbf{J}}

\newcommand{\bz}{\mathbf{z}}
\newcommand{\by}{\mathbf{y}}

\newcommand{\bX}{\mathbf{X}}

\newtheorem{theorem}{Theorem}
\newtheorem{corollary}{Corollary}
\newtheorem{claim}{Claim}
\newtheorem{lemma}{Lemma}

\usepackage{xcolor}

\title{New Lower Bounds for Adaptive Tolerant Junta Testing}
\author{Xi Chen\footnote{Email: \texttt{xichen@cs.columbia.edu}}
\\ \textsl{Columbia University} \and
Shyamal Patel\thanks{Email: \texttt{shyamalpatelb@gmail.com}} \\ \textsl{Columbia University} }
\date{}
\begin{document}

\maketitle

\input{Sections/Intro.tex}
\input{Sections/Lowerbound.tex}
\input{Sections/TestingGaps.tex}
\input{Sections/Barriers.tex}

\bibliographystyle{alpha}
\bibliography{references}

\input{Sections/AppendixA.tex}
\input{Sections/AppendixB.tex}

\end{document}

%% file: Sections/Intro.tex

\begin{abstract}
We prove a $k^{-\Omega(\log(\varepsilon_2 - \varepsilon_1))}$ lower bound for adaptively testing whether a Boolean function is $\varepsilon_1$-close to or $\varepsilon_2$-far from $k$-juntas. 
Our results provide the first superpolynomial separation between tolerant and non-tolerant testing for a natural property of boolean functions under the adaptive setting. Furthermore, our techniques generalize to show that adaptively testing whether a function is $\varepsilon_1$-close to a $k$-junta or $\varepsilon_2$-far from $(k + o(k))$-juntas cannot be done with $\textsf{poly} (k, (\varepsilon_2 - \varepsilon_1)^{-1})$ queries. This is in contrast to an algorithm by Iyer, Tal and Whitmeyer [CCC 2021] which uses $\textsf{poly} (k, (\varepsilon_2 - \varepsilon_1)^{-1})$ queries to test whether a function is  $\varepsilon_1$-close to a $k$-junta or $\varepsilon_2$-far from $O(k/(\varepsilon_2-\varepsilon_1)^2)$-juntas.\end{abstract}
\newpage

\section{Introduction}

\noindent\textbf{Junta Testing.} We say  a Boolean function $f: \{0,1\}^n \rightarrow \{0,1\}$ is a \emph{$k$-junta} if it only depends on $k$ of its input variables. As one of the most fundamental classes of Boolean functions, 
juntas have received significant attention during the past few decades in extensive areas such as learning theory \cite{blum1994relevant, blum1997selection, mossel2003learning, valiant2015finding} 
  (where juntas are used to model learning concepts in the 
  presence of many irrelevant features) and analysis of Boolean functions \cite{ODonell2014analysis}
  (where juntas~are~used~as~good approximations of other classes of Boolean functions). Juntas have also been studied intensively in property testing.
Under the standard model,
  the goal of the junta testing problem is to understand how many queries are~needed by a randomized algorithm to decide whether a function $f: \{0,1\}^n \rightarrow \{0,1\}$ is~a $k$-junta or $\eps$-far from $k$-juntas, where we say $f$ is $\eps$-far from $k$-juntas if $f$ and $g$ disagree on at least $\eps 2^n$ entries for every $k$-junta $g: \{0,1\}^n \rightarrow \{0,1\}$. Testing juntas, in particular, was initially of interest because of its connections to testing long-codes (which is related to PCPs) \cite{bellare1998free, parnas2002testing}, but it is also a very natural question and highly motivated by feature selection in machine learning.

This problem has been settled now under both the adaptive and non-adaptive settings.\footnote{A 
  testing algorithm is adaptive if its queries can depend on results from previous queries.}
After a~sequence of work \cite{parnas2002testing, chockler2004lower, fischer2004testing, blais2008improved, blais2009testing, buhrman2012non, servedio2015adaptivity, chen2018settling, sauglam2018near}, starting with \cite{parnas2002testing}, it was shown that $\tilde{\Theta}(k/\eps)$ queries are necessary \cite{chockler2004lower, sauglam2018near} and sufficient \cite{blais2009testing} for adaptive junta testing; 
  $\tilde{\Theta}(k^{3/2}/\eps)$ queries are 
  necessary \cite{chen2018settling} and sufficient \cite{blais2008improved} for non-adaptive 
  algorithms.
Note that all the bounds are independent of $n$, the number of variables in the function, which can be
  much larger than $k$.
  \medskip
  
\noindent\textbf{Tolerant Testing.}
A drawback of the standard testing model is that the algorithm is allowed~to reject functions 
  that are very close to having the property.
Indeed all aforementioned junta testing algorithms would reject a function immediately once
  it has been found to have more than $k$ relevant variables, no matter how big their influences are in the function.
This makes algorithms under the standard model less applicable in more realistic scenarios
  that arise from experimental design and data analysis,
 where 
  even a correct instance with the desired property may be subject to a small amount of noise.

To address this question,
Parnas, Ron, and Rubinfeld \cite{parnas2006tolerant} introduced \emph{tolerant testing},~a natural generalization of the standard testing model. For the tolerant testing of juntas, the goal of a $(k,\eps_1, \eps_2)$-tester, for some $0\le \eps_1 < \eps_2 < 1$, is to tell whether a given function $f:\{0,1\}^n\rightarrow \{0,1\}$ is $\eps_1$-close to a $k$-junta or is $\eps_2$-far from $k$-juntas, where we say $f$ is $\eps_1$-close to a $k$-junta $g$ if $f$ and $g$ disagree on no more than $\eps_1 2^n$ entries.
So the standard model corresponds to the case of $\eps_1= 0$. 
In general, tolerant testing of a property can be much more challenging than its standard testing, as intuitively seeing a single violation of the property is no longer sufficient to reject the input and the tester needs to estimate the distance of the input to the property.\footnote{It was shown in \cite{parnas2006tolerant} that distance approximation and (fully) tolerant testing are equivalent up to a $\log(1/\eps)$ factor in the query complexity.}

While it is believed that tolerant testing is much harder 
  and current best upper bounds for well studied properties of Boolean functions such as juntas (which we review next)
  and monotonicity all remain exponential in their counterparts under the non-tolerant model
  \cite{iyer2021junta, fattal2010approximating},
there is no known superpolynomial separation between tolerant and non-tolerant testings
  for natural properties of Boolean functions.\footnote{Using PCPs, \cite{fischer2005tolerant} showed the existence
  of a class of Boolean functions that has a strong separation 
  between the two models.}  
\emph{Our work proves the first such superpolynomial separation, using junta testing.}  
We also believe our approach may be fruitful in proving tolerant testing lower bounds for other properties such as monotonicity.

\noindent\textbf{Previous Results on Tolerant Junta Testing.}
Indeed  testing juntas under the tolerant~model turns out to be challenging. Even after much effort,
  there remains an exponential  
  gap in 
  our understanding of its query complexity.
After a sequence of work \cite{chakraborty2012junto, blais2019tolerant, de2019junta, iyer2021junta}, the best algorithm up to date \cite{iyer2021junta} (which is highly adaptive) still needs $\exp ( \tilde{O} ( \sqrt{ {k}/{(\eps_2 - \eps_1)}} ) )$ queries. 
And there has been no progress on the lower bound side: the current best 
  lower bound for adaptive testers remains to be the $\tilde{\Omega}(k/\eps)$ lower bound from the standard testing model.  

There has been more success on lower bounds when the tester is non-adaptive.
 \cite{levi2019lower} showed that any non-adaptive algorithm requires $\tilde{\Omega}(k^2)$ queries, for some constants $0<\eps_1 < \eps_2$. 
Later, for $\eps_1 = O(1/k^{1-\eta})$ and $\eps_2 = O(1/\sqrt{k})$,
  the bound was improved to $\smash{2^{k^{\eta}}}$ for any $0 < \eta < 1/2$ \cite{pallavoor2022approximating}.\medskip

\noindent\textbf{Relaxed Tolerant Junta Testing.} 
 Given the state of the art,
   the~following relaxed, easier model has been considered in the literature \cite{fischer2004testing,blais2011property,ron2013approximating,de2019junta,iyer2021junta}, with the hope~of developing  
   algorithms that are significantly more efficient than the best tolerant algorithm to date: For some $k\le k'$ and $\eps_1<\eps_2$,
the goal of a $(k, k', \eps_1, \eps_2)$-tester is to decide whether a given function is  $\eps_1$-close to a $k$-junta or is $\eps_2$-far from  $k'$-juntas. 

The relaxed model was first posed by Fischer et al. in \cite{fischer2004testing} under the 
	  non-tolerant setting (with $\eps_1=0$), who asked whether a $(k, 2k, 0, \eps)$-tester still requires $\Omega(k)$ queries. Blais~et~al.~\cite{blais2011property} proved that any $(k, k + O(\sqrt{k}), 0, \eps)$-tester must make $\Omega(k)$ queries for $\eps = \Theta(1)$. This was improved by~Ron~and Tsur \cite{ron2013approximating} to an $\tilde{\Omega}(k)$ lower bound for $(k, 2k, 0, \eps)$-testers for some constant $\eps$.
	
	As before, under the relaxed, tolerant setting ($\eps_1>0$), the best known lower bound remains~to be the $\tilde{\Omega}(k)$ lower bound of \cite{ron2013approximating} inherited from the non-tolerant setting. In contrast, much more efficient algorithms are known when $k'$ is sufficiently larger than $k$. First~Blais~et~al.~\cite{blais2019tolerant}~gave a $(k, 4k, \eps/16, \eps)$-tester that makes $\poly(k, \eps^{-1})$ queries.   De et al.~\cite{de2019junta} removed~the restriction of $\eps_2 / \eps_1 \ge 16$ and proved that $\poly(k,  ({\eps_2 - \eps_1})^{-1})$ queries suffice for a $(k, O(k^2/(\eps_2-\eps_1)^2), \eps_1, \eps_2)$-tester. More recently, Iyer et al. in \cite{iyer2021junta} obtained a $(k, O(k/(\eps_2-\eps_1)^2), \eps_1, \eps_2)$-tester that makes $\poly(k, ({\eps_2 - \eps_1})^{-1})$ queries.\medskip

\noindent\textbf{Our Contribution.} Obtaining stronger adaptive lower bounds for tolerant junta testing has~been an important open problem in the property testing of Boolean functions (see e.g., the open problem posed in \cite{iyer2021junta}).
In this paper, we make progress on this question by giving a superpolynomial separation between
  tolerant and non-tolerant adaptive junta testing:

\begin{theorem}\label{theo1}
Let $\eps_1,\eps_2 :0.01 \leq \eps_1 < \eps_2 \leq 0.49$ be two parameters such that $\eps_2 - \eps_1 \geq 2^{-O(k^{.99})}$.\footnote{Both constants $0.01$ and $0.49$ are arbitrary; any constants that are positive and strictly less than $1/2$ would work.}
Then any $(k, \eps_1, \eps_2)$-tolerant junta tester requires 
$k^{-\Omega(\log(\eps_2  - \eps_1))}$ queries.
\end{theorem}

Theorem \ref{theo1} rules out the possibility of any tolerant junta tester that makes $\poly(k)\cdot F(\eps_1, \eps_2)$ many queries, where $F$ is an arbitrary function.
In particular, this means that there is no tolerant junta tester that makes $\poly(k)$ queries for 
  all constants $\eps_1,\eps_2: 0<\eps_1<\eps_2<1$.
Moreover,
  when $\eps_2-\eps_1$ is polynomially small (e.g., $\eps_1=1/3$ and $\eps_2=1/3+1/k^a$ for any constant $a>0$),
  Theorem \ref{theo1} gives a lower bound of $k^{\Omega(\log k)}$ for adaptive tolerant junta testing.

We remark that our proof naturally extends the linear lower bound of Chockler and Gutfreund \cite{chockler2004lower}. Namely, we prove Theorem \ref{theo1} by showing that determining whether a function on $(k + \ell)$ variables is $\eps_2$-far or $\eps_1$-close to a $k$-junta requires $k^{\Omega(\ell)}$ queries. Our main technical insight is that utilizing $k$-wise independence in the construction of hard instances essentially allows us to assume that the tester is non-adaptive in the lower bound proof. 

Next, by making some slight modifications to our construction, we can extend our techniques to the relaxed tolerant junta testing setting. Namely, we prove that

	\begin{theorem}
	\label{thm:tolerant-gap-lower-bound}
Let $\eps_1,\eps_2,k$ and $\gamma$ be parameters such that $0.01 \leq \eps_1 < \eps_2 \leq 0.49$, 
$\eps_2 - \eps_1 \geq 2^{-k^{0.1}}$ and 
  $\gamma \ge 1/k$ but is sufficiently small. 
  Then any $(k, (1 + \gamma) k, \eps_1, \eps_2)$-tester has query complexity $$\left( \frac{1}{\gamma} \right)^{\Omega(-\log(\eps_2 - \eps_1))}.$$ 
	\end{theorem}

Setting $\eps_2 - \eps_1 = 1/{k}$, we get a superpolynomial lower bound (in $k$) whenever $\gamma = o(1)$. Hence Theorem \ref{thm:tolerant-gap-lower-bound} rules out the possibility of any $(k, k + o(k), \eps_1, \eps_2)$-tester that makes   $\poly(k,  ({\eps_2 - \eps_1})^{-1})$ queries,
  which complements the need of a sufficiently large gap between $k'$ and $k$ to obtain
  testers with  $\poly(k,  ({\eps_2 - \eps_1})^{-1})$ queries in \cite{blais2019tolerant,de2019junta,iyer2021junta}. \medskip

\noindent\textbf{Notation.}
We start with some simple notation. Given a set $S \subseteq [n]$ and a binary string $x \in \{0,1\}^n$, we let $x|_S$ be the restriction of $x$ to the coordinates in $S$. Similarly, given a string $x \in \{0,1\}^S$ and $y \in \{0,1\}^T$ for disjoint sets $S, T \subseteq [n]$ we let $x \sqcup y$ denote the string $z \in \{0,1\}^{S \cup T}$ with $z|_S = x$ and $z|_T = y$. For a set $S \subseteq [n]$ and $x \in \{0,1\}^n$, we take $x^{\oplus S}$ to denote the string $y \in \{0,1\}^n$ with $y|_{[n] \setminus S} = x|_{[n] \setminus S}$ and $y_s \not = x_s$ for all $s \in S$. Finally, we let $x_{-i} = x_1 x_2 ... x_{i-1} x_{i+1} ... x_n$.

	$B_n(r)$ will denote the hamming ball of radius $r$ centered at $0^n$ and $B(x,r)$ will be the hamming ball around a point $x \in \{0,1\}^n$ of radius $r$.

	We let $\mathbb{I}_E$ denote the indicator variable for $E$. Additionally, $\maj$ denotes the majority function, breaking ties arbitrarily. 

%% file: Sections/Lowerbound.tex

\section{Lower Bounds for Tolerant Junta Testing}
We will actually prove a stronger lower bound which allows for non-constant $\eps_1$ and $\eps_2$.

\begin{theorem}
\label{thm:tolerant-lower-bound}
For any $0 < \eps_1 < \eps_2 \leq 0.49$ and integer $k$ with $\eps_2 - \eps_1 \geq 2^{-O(k^{.99})}$ and 
\[(1 - \eps_1/\eps_2)^{-1} \geq \log^{O(1)} (1/\eps_2),\] 
any $(k, \eps_1, \eps_2)$-tolerant junta tester must make at least $k^{-\Omega(\log(1 - \eps_1/\eps_2))}$ queries.
\end{theorem}

We'll assume that $1-\eps_1/\eps_2 $ is sufficiently small. Note that we can always take the constant in the $\Omega$ notation to be sufficiently small, so that when $1-\eps_1/\eps_2 $ is large the theorem only gives an $\Omega(k)$ lower bound, which follows from \cite{chockler2004lower}. 

The key insight for our lower bound is the following observation: 

\begin{lemma}
\label{lem:k-wise-distance}
Let $n$ be an integer divisible by $4$, $\calD$ be the uniform distribution over $\{0,1\}^n$, and $\calD'$ be the uniform distribution over $ \{x \in \{0,1\}^n : \oplus_{i=1}^n x_i = 1\}$. Then
	\[\E_{\bx \sim \calD'} \big [ \dist(\bx, \{0^n, 1^n\}) \big] \leq \E_{\bx \sim \calD} \big [ \dist(\bx, \{0^n, 1^n\}) \big] - \frac{1}{n}, \]
where $0^n$ and $1^n$ are the all zeros and all ones strings respectively and $\dist$ is the hamming distance.
\end{lemma}

We remark that the $\frac{1}{n}$ term in the lemma is not tight, but will suffice for our purposes. The proof uses a coupling argument; we defer it to the end of the section.

Our construction will be parametrized by integers $\ell$ and $r$ as well as a value $p \in [0, \frac{1}{2})$, which will be specified at the end (in the proof of 
Theorem \ref{thm:tolerant-lower-bound}) by $k,\eps_1$ and $\eps_2$ and in particular, we will make sure that $r<k$. 
It may help the reader to set $r = 0$ and $p = 0$, which will lead to a lower bound construction that works for some $\eps_2 = \frac{1}{2} - o_\ell(1)$ and $\eps_1 = \eps_2 - 2^{-O(\ell)}$ (in which case $\calD_{NO}$ is simply the uniform distribution over all functions). Setting $r$ and $p$ appropriately (as we do later in the proof of Theorem \ref{thm:tolerant-lower-bound}) can help us shift $\eps_1$ and $\eps_2$ to where we want.

We will proceed by Yao's principle and give two distributions $\calD_{YES}$ and $\calD_{NO}$, both over boolean functions from $\{0,1\}^n$ to $\{0,1\}$ with $n := k + \ell$.
We show that they are  $\eps_1$-close and $\eps_2$-far from $k$-juntas, respectively, and show that it is difficult for any deterministic algorithm to distinguish them with few queries.
\begin{flushleft}
\begin{enumerate}
    	\item[] $\calD_{NO}$: 
	A boolean function $\bff:\{0,1\}^n\rightarrow\{0,1\}$ is drawn as follows:
	\begin{enumerate}
	\item For each $x\in \{0,1\}^n$ with $x|_{[r]} \not = 0^r$, independently draw $\bff(x)$ from $\bern(p)$.
	\item For each $x\in \{0,1\}^n$ with $x|_{[r]} = 0^r$, draw $\bff(x)$ uniformly
	and independently at random.
    \end{enumerate}
	    \item[] $\calD_{YES}$: A boolean function $\bff:\{0,1\}^n\rightarrow \{0,1\}$ is drawn as follows:
	\begin{enumerate}
	    \item Randomly choose a set $\bJ \subseteq \{r+1, ..., n\}$ of size $\ell$. 
	    \item For each $x\in \{0,1\}^n$ with $x|_{[r]} \not = 0^r$, independently draw $\bff(x)$ from $\bern(p)$. 
	    \item For each $x\in \{0,1\}^n$ with $x|_{[r]} = 0^r$ and 
	    $x |_\bJ \not = 1^\ell$, draw $\bff(x)$ uniformly and independently at random. 		\item For each $x\in \{0,1\}^n$ with $x|_{[r]} = 0^r$ and $x|_\bJ = 1^\ell$,
	    set $\bff(x)$ to be such that
    	    \[\bigoplus_{y \in \{0,1\}^\bJ} \bff(x|_{[n] \setminus \bJ} \sqcup y) = 1.\]
	\end{enumerate}
\end{enumerate}
\end{flushleft}	
The key property of the two distributions is that to distinguish $\calD_{YES}$ and $\calD_{NO}$ we must at least query a pair of points $x, x^{\oplus \bJ}$ for some $x$. (In fact, we must query, for some $x$, $x^{\oplus S}$ for all $S \subseteq \bJ$ to get any evidence, but this is more than we need for the lower bound.)

\begin{lemma}
\label{lem:almost-same-distribution}
Consider a set of points $x^{(1)}, ..., x^{(m)} \in \{0,1\}^n$ and $y_1,\ldots,y_m\in \{0,1\}$.
Then
	\[\Pr_{\bff, \bJ \sim \calD_{YES}} \left[\forall i\hspace{2mm}  \bff(x^{(i)}) = y_i \mid \forall i,j  \quad x^{(i)} \not = (x^{(j)})^{\oplus \bJ} \right] = \Pr_{\bff \sim \calD_{NO}} \left[\forall i \hspace{2mm} \bff(x^{(i)}) = y_i \right]\]
\end{lemma}
	
\begin{proof}
Consider a $J \subseteq \{r+1, ..., n\}$ such that $x^{(i)} \not = (x^{(j)})^{\oplus J}$ for all $i,j\in [m]$. We will show that
		\[\Pr_{\bff, \bJ \sim \calD_{YES}} \left[\forall i\hspace{2mm}  \bff(x^{(i)}) = y_i \land \bJ = J \right] = \Pr_{\bff \sim \calD_{NO}} \left[\forall i \hspace{2mm} \bff(x^{(i)}) = y_i \right] \cdot  \Pr_{\bff, \bJ \sim \calD_{YES}}\big[\bJ = J\big]. \]
	Indeed, without loss of generality let $x^{(1)}, ..., x^{(a)}$ be such that $x^{(i)}|_{[r]} \not = 0^r$ and $x^{(a+1)} ... x^{(m)}$ be such that $x^{(i)}|_{[r]} = 0^r$. Now note that for any $\rho \in \{0,1\}^{[n] \setminus \bJ}$, the following $2^\ell$ bits
		\[\{\bff(\rho \sqcup z): z \in \{0,1\}^\ell \}\]
	are $(2^{\ell}-1)$-wise independent. So it follows that the events $\bff(x^{(i)}) = y_i$ are all independent and 
		\[\Pr_{\bff, \bJ \sim \calD_{YES}} \left[\forall i\hspace{2mm}  \bff(x^{(i)}) = y_i \mid \bJ = J \right] = \left( \prod_{i=1}^a p^{y_i} (1 - p)^{1 - y_i} \right) \cdot \frac{1}{2^{m-a}} = \Pr_{\bff \sim \calD_{NO}} \left[\forall i \hspace{2mm} \bff(x^{(i)}) = y_i \right]  \]
This finishes the proof of the lemma.
\end{proof}

It now follows that $\calD_{YES}$ and $\calD_{NO}$ are hard to distinguish.

\begin{lemma}
\label{lem:query-lower-bound}
	Any deterministic algorithm $\ALG$ that distinguishes between $\calD_{YES}$ and $\calD_{NO}$ with probability at least $2/3$, i.e., $\Pr_{\bff \sim \calD_{YES}} [\ALG(\bff) \text{ accepts}] \geq 2/3$ and $\Pr_{\bff \sim \calD_{NO}} [\ALG(\bff) \text{ rejects}] \geq 2/3$, must make at least $\Omega \left( \sqrt{\binom{k-r}{\ell}} \right)$ queries. 
\end{lemma}

\begin{proof}
Towards a contradiction, suppose $\ALG$ distinguishes the distributions and makes at most $Q = \frac{1}{10} \sqrt{\binom{k -r}{\ell}}$ queries. Since $\ALG$ is deterministic, it corresponds to a decision tree. But now observe that for a particular path $p$ of the decision tree, we have that
			\[\Pr_{\bff \sim \calD_{NO}}[\ALG(\bff) \text{ follows } p] = \Pr_{\bff, \bJ \sim \calD_{YES}}[\ALG(\bff) \text{ follows } p |  p \text{ doesn't query a pair $x, x^{\oplus \bJ}$} ]. \]
		by Lemma \ref{lem:almost-same-distribution}. We then conclude that
		\[\Pr_{\bff \sim \calD_{NO}}[\ALG(\bff) \text{ follows } p] \leq \frac{\Pr_{\bff \sim \calD_{YES}}[\ALG(\bff) \text{ follows } p]}{\Pr_{\bff, \bJ \sim \calD_{YES}} [p \text{ doesn't query a pair $x, x^{\oplus \bJ}$}]}. \]
		Now observe		
			\[ \Pr_{\bff, \bJ \sim \calD_{YES}} [p \text{ doesn't query a pair $x, x^{\oplus \bJ}$}] \geq 1 - \frac{Q^2}{\binom{k - r}{\ell}} \geq .99. \]
		Thus
			\[ \Pr_{\bff \sim \calD_{NO}}[\ALG(\bff) \text{ follows } p] \leq 1.02 \cdot \Pr_{\bff \sim \calD_{YES}}[\ALG(\bff) \text{ follows } p] \]
		 Summing over all rejecting paths we conclude that
			\[\Pr_{\bff \sim \calD_{NO}}[\ALG(\bff) \text{ rejects} ] \leq 1.02 \cdot \Pr_{\bff \sim \calD_{YES}}[\ALG(\bff) \text{ rejects}] < 1/2, \]
		a contradiction. So any tester must make $\Omega \left(\sqrt{\binom{k - r}{\ell}} \right)$ queries as claimed.	
\end{proof}

Next, we need to understand exactly how far functions in $\calD_{NO}$ are from being $k$-juntas. To do so, we'll need the following helper lemma

\begin{lemma}
\label{lem:unbalanced-distance-to-constant}
	Let $\bx_1, ..., \bx_n$ be independent boolean valued random variables such that $\bx_i \sim \bern(p_i)$ for constants $p_i \in [0,1/2]$. If $\frac{1}{2} - \frac{1}{n} \sum_{i = 1}^n p_i \geq \delta$ then 
		\[\left|\E_{\bx_1, ..., \bx_n} \left[ \dist(\bx_1\bx_2...\bx_n, \{0^n, 1^n\}) \right] - \sum_{i = 1}^n p_i \right| \leq n e^{- n\delta^2/3} \]
\end{lemma}

\begin{proof}
We start by proving that it is unlikely that $\bx_1, ..., \bx_n$ is closer $1^n$ than $0^n$. Indeed, the probability of this occurring is
	\[\Pr_{\bx_1, ..., \bx_n} \left[\sum_{i=1}^n \bx_i > n/2 \right] \leq e^{- n\delta^2/3}\]
by a Hoeffding bound. So we conclude that
	\[\bigg| \E_{\bx_1, ..., \bx_n} \left[ \dist(\bx_1...\bx_n, \{0^n, 1^n\}) \right] - \E_{\bx_1, ..., \bx_n} \left[ \dist(\bx_1...\bx_n, 0^n) \right] \bigg| \leq n e^{- n \delta^2/3}.\]
The result now follows as 
	\[\E[\dist(\bx_1, ..., \bx_n, 0^n)] = \sum_{i=1}^n p_i.\]
This finishes the proof of the lemma.
\end{proof}

Let $\Delta_{t}$ be defined as follows:
	\[\Delta_{t} := \frac{1}{t}\cdot \E_{\bx_1 ... \bx_{t} \sim \bern(\frac{1}{2})} \big[\dist(\bx_1, ..., \bx_{t}, \{0^{t}, 1^{t}\})\big]. \]

\begin{lemma}
\label{lem:no-distribution-distance}
Suppose that $\ell \leq k$, then 	
	\[\Pr_{\bff \sim \calD_{NO}} \left[ \dist(\bff, \calJ_{k}) \leq  p (1-2^{-r}) + \Delta_{2^\ell} \cdot 2^{-r} - e^{-2^\ell \cdot (1/2 - p)^2/12} - 2^{-k/3} \right] = o_k(1)\]
where $\calJ_{k}$ denotes the class of $k$-juntas.
\end{lemma}

\begin{proof}
Fix a set $S \subseteq [n]$ of size $k$. Let $\calJ_S$ denote the set of $k$-juntas on $S$. Now take $I = S \cap [r]$ and note that
	\[\dist(\bff, \calJ_S) = \frac{1}{2^k} \sum_{\rho \in \{0,1\}^S} \sum_{y \in \{0,1\}^{[n] \setminus S}} \frac{1}{2^\ell} \left|\bff(\rho \sqcup y) - \maj(\{\bff(\rho \sqcup y): y \in \{0,1\}^{[n] \setminus S} \}) \right| \]
We will lower bound $\E_{\bff}[\dist(\bff, \calJ_S)]$. Consider some fixed $\rho \in \{0,1\}^S$. If $\rho|_I \not = 0$ then
	\[ \frac{1}{2^\ell} \cdot \E_{\bff} \left[ \sum_{y \in \{0,1\}^{[n] \setminus S}} \bff(\rho \sqcup y) \right] = p. \]
So by Lemma \ref{lem:unbalanced-distance-to-constant} it follows that
	\[\E_{\bff} \left[ \frac{1}{2^\ell} \sum_{y \in \{0,1\}^{[n] \setminus S}} \left|\bff(\rho \sqcup y) - \maj(\{\bff(\rho \sqcup y): y \in \{0,1\}^{[n] \setminus S} \}) \right| \right] \geq p - e^{-2^\ell \cdot (1/2-p)^2/3}. \]
So now suppose that $\rho_I = 0^{|I|}$. We take cases on $|I|$. First suppose that $|I| < r$.  Then
	\[ \frac{1}{2^\ell} \cdot \E_{\bff} \left[ \sum_{y \in \{0,1\}^{[n] \setminus S}} \bff(\rho \sqcup y) \right] = p(1-2^{|I|-r}) + \frac{1}{2} \cdot 2^{|I|-r} \]
as the $2^{|I|-r}$ fraction of $x$ values with $x|_{[r]} = 0^r$ are distributed according to $\bern(\frac{1}{2})$ and the rest are distributed according to $\bern(p)$. Using Lemma \ref{lem:unbalanced-distance-to-constant} along with the fact that $r - |I| \geq 1$, we get
		\[\E_{\bff} \left[ \frac{1}{2^\ell} \sum_{y \in \{0,1\}^{[n] \setminus S}} \left|\bff(\rho \sqcup y) - \maj(\{\bff(\rho \sqcup y): y \in \{0,1\}^{[n] \setminus S} \}) \right| \right] \geq p(1-2^{|I|-r}) + \frac{1}{2} \cdot 2^{|I|-r} - e^{-2^\ell \cdot (1/2 - p)^2/12}. \]

All together we see that in this case
	\begin{align*}
		\E_{\bff} [\dist(\bff, \calJ_S)] &\geq (1 - 2^{-|I|}) p + 2^{-|I|} \left (p(1-2^{|I|-r}) + \frac{1}{2} \cdot 2^{|I|-r} \right) - e^{-2^\ell \cdot (1/2 - p))^2/12} \\
		&= p(1-2^{-r}) + \frac{1}{2^{r+1}} - e^{-2^\ell \cdot (1/2 - p))^2/12}
	\end{align*}

Now suppose that $|I| = r$. It then follows that all the entries are $\bern(\frac{1}{2})$ and 
	\[\frac{1}{2^\ell} \cdot \E_{\bff} \left[ \sum_{y \in \{0,1\}^{[n] \setminus S}} \left|\bff(\rho \sqcup y) - \maj(\{\bff(\rho \sqcup y): y \in \{0,1\}^{[n] \setminus S} \}) \right| \right] = \Delta_{2^\ell}\]
So we conclude that in this case
	\[\E_{\bff} \left[ \dist(\bff, \calJ_S) \right] \geq p (1-2^{-r}) + \Delta_{2^\ell} \cdot 2^{-r} - e^{-2^\ell \cdot (1/2 - p)^2/12} \]

Since $\Delta_{2^\ell} < \frac{1}{2}$, it follows that for any set $S \subseteq [n]$	
	\[\E_{\bff} \left[ \dist(\bff, \calJ_S) \right] \geq p (1-2^{-r}) + \Delta_{2^\ell} \cdot 2^{-r} - e^{-2^\ell \cdot (1/2 - p)^2/12}.\]

We now note that $\dist(\bff, \calJ_S)$ is the average of independent random variables in $[0,1]$. Namely, it is the average the random variables
	\[\bz_\rho := \sum_{y \in \{0,1\}^{[n] \setminus S}} \frac{1}{2^\ell} \left|\bff(\rho \sqcup y) - \maj(\{\bff(\rho \sqcup y): y \in \{0,1\}^{[n] \setminus S} \}) \right|\]
Thus, by Hoeffding's inequality we have that
	\[\Pr_{\bff \sim \calD_{NO}} \left [\dist(\bff, \calJ_S) \leq \E[\dist(\bff, \calJ_S)] - 2^{-k} t \right ] \leq e^{-2t^2/2^k} \]
Taking $t$ to be $2^{2k/3}$, it follows that
	\[\Pr_{\bff \sim \calD_{NO}} \left [\dist(\bff, \calJ_S) \leq \E[\dist(\bff, \calJ_S)] - 2^{-k/3} \right ] \leq e^{-2 \cdot 2^{k/3}} \]
We can now take a union bound over all subsets to conclude that
	\[\Pr_{\bff \sim \calD_{NO}} \left [\dist(\bff, \calJ_k) \leq p (1-2^{-r}) + \Delta_{2^\ell} \cdot 2^{-r} - e^{-2^\ell \cdot (1/2 - p)^2/12} -2^{-k/3}\right ]  \leq e^{-2 \cdot 2^{k/3}} (k+\ell)^k = o_k(1)\]
\end{proof}

We will also need an analogous result for functions drawn from $\calD_{YES}$.

\begin{lemma}
\label{lem:d-yes-distance} 	
Suppose that $\ell \geq 2$, then
	\[\Pr_{\bff \sim \calD_{YES}} \left[ \dist(\bff, \calJ_{k}) \geq p (1-2^{-r}) + (\Delta_{2^\ell} - 2^{-2 \ell})) \cdot 2^{-r} + e^{-2^\ell \cdot (1/2 - p)^2/12} + 2^{-k/3} \right] = o_k(1)\]
where $\calJ_{k}$ denotes the class of $k$-juntas.
\end{lemma}

\begin{proof}
	The proof follows very similarly to Lemma \ref{lem:no-distribution-distance}.	
	\begin{align*}
		&\hspace{-0.4cm}\E_{\bff, \bJ \sim \calD_{YES}}[\dist(\bff, \calJ_{[n] \setminus \bJ})]\\ &= \frac{1}{2^k} \E_{\bff, \bJ \sim \calD_{YES}} \left[ \sum_{\rho \in \{0,1\}^{[n] \setminus \bJ}}  \sum_{y \in \{0,1\}^\bJ} \frac{1}{2^\ell} \left|\bff(\rho \sqcup y) - \maj(\{\bff(\rho \sqcup y): y \in \{0,1\}^{\bJ} \}) \right| \right]	\\
		& \leq p (1-2^{-r}) + e^{-2^\ell \cdot (1/2 - p)^2/12} + \left(\Delta_{2^\ell} - 2^{-2\ell}\right)\cdot 2^{-r} 
	\end{align*}
	where $p (1-2^{-r}) + e^{-2^\ell \cdot (1/2 - p)^2/12}$ corresponds to the contribution from terms with $\rho|_{[r]} \not = 0^r$ and $\left(\Delta_{2^\ell} - 2^{-2\ell}\right)\cdot 2^{-r}$ bounds the expectation for terms with $\rho|_{[r]} = 0^r$ by Lemma \ref{lem:k-wise-distance}. Applying a Hoeffding bound now gives us the desired result.
\end{proof}

We can now put everything together and prove the lower bound.

\begin{proof}[Proof of Theorem \ref{thm:tolerant-lower-bound}]
	It only remains to set the parameters so any $(k, \eps_1, \eps_2)$ tolerant tester is forced to distinguish the two distributions. We let $\ell = \lfloor \frac{-1}{10} \log(1-\eps_1 \eps_2^{-1}) \rfloor$. We then take $r$ to be the smallest integer such that
		\[\Delta_{2^\ell} 2^{-r} \leq \eps_2 + 2^{-k/3} + e^{-2^\ell \cdot (0.009)^2/12} \]
	and set $p$ to satisfy
		\[ \eps_2 + 2^{-k/3} + e^{-2^\ell \cdot (0.009)^2/12} = p (1-2^{-r}) + \Delta_{2^\ell} \cdot 2^{-r} .\]
	We assume that $\ell$ and $k$ are sufficiently large such that $2^{-k/3} \leq .005$, $e^{-2^\ell \cdot (0.009)^2/12} \leq .005$, and $\Delta_{2^\ell} \geq .491$. Note that these together imply that $r \geq 0$ and $0 \leq p \leq .491$. 	Moreover, note that by the minimality of $r$, we have that
		\[\Delta_{2^\ell} 2^{-r + 1} \geq \eps_2.\] 
	which implies that $2^{-r} \geq \eps_2$ and thus assuming $\eps_2-\eps_1 \geq 2^{-k/6}$, $r\le k/6$ as promised.
	
	 To see that a $(k, \eps_1, \eps_2)$ tolerant tester must distinguish $\calD_{YES}$ and $\calD_{NO}$, first observe that by Lemma \ref{lem:no-distribution-distance} and our choice of parameters, we have that functions in $\calD_{NO}$ are at least $\eps_2$-far from the set of $k$-juntas with high probability. On the other hand, by Lemma \ref{lem:d-yes-distance}, functions from $\calD_{YES}$ are with high probability 
	 	\[p (1-2^{-r}) + (\Delta_{2^\ell} - 2^{-2 \ell})) \cdot 2^{-r} + e^{-2^\ell \cdot (.009)^2/12} + 2^{-k/3}\]
	 close to some $k$-junta. We'll show this is at most $\eps_1$ under our choice of parameters. To see this first note that
	 	
		\[e^{-2^{\ell} (.009)^2/12} \leq e^{-(1-\eps_1 \eps_2^{-1})^{-1/10} (.009)^2/24} \leq e^{-(1-\eps_1 \eps_2^{-1})^{-1/11}}\]
	where the second inequality assumes $(1 - \eps_1 \eps_2^{-1})^{-1}$ is sufficiently large. Now we assume that $(1-\eps_1 \eps_2^{-1})^{-1} \geq \log^{12}(\frac{1}{\eps_2})$, which we took as a hypothesis in the theorem, and observe
		\[-2\log(\eps_2 - \eps_1) = -2\log(\eps_2) - 2\log(1 - \eps_1 \eps_2^{-1}) \leq 2(1 - \eps_1 \eps_2^{-1})^{-1/12}  + 2(1 - \eps_1 \eps_2^{-1})^{-1/12} \leq (1 - \eps_1 \eps_2^{-1})^{-1/11}\]
	for $(1-\eps_1 \eps_2^{-1})^{-1}$ sufficiently large. Thus, 
		\[e^{-(1-\eps_1 \eps_2^{-1})^{-1/11}} \leq e^{2 \log(\eps_2 - \eps_1)} = (\eps_2 - \eps_1)^2. \]
		
	Plugging these into our expression for the distance of functions in $\calD_{YES}$ to the set of $k$-juntas yields that they are
	\begin{align*}
		\eps_2 + 4 (\eps_2 - \eps_1)^2 - 2^{\frac{1}{5} \log(1-\eps_1 \eps_2^{-1})} \eps_2 & \leq \eps_2 + 4 (\eps_2 - \eps_1)^2 - (1-\eps_1 \eps_2^{-1})^{\frac{1}{5}} \eps_2 \\
		&\leq \eps_2 + 4 (\eps_2 - \eps_1)^2 - 10 (1-\eps_1 \eps_2^{-1}) \eps_2 \leq \eps_1
	\end{align*}

		close to being a $k$-junta with high probability, again assuming $(1-\eps_1\eps_2^{-1})^{-1}$ is sufficiently large.
	
	Thus any tolerant tester must distinguish the two distributions and the result follows from Lemma \ref{lem:query-lower-bound}.
\end{proof}

Finally, we conclude by proving Lemma \ref{lem:k-wise-distance}.

\begin{proof}[Proof of Lemma \ref{lem:k-wise-distance}]
	We'll prove the statement by a coupling argument. Note that we can sample from $\calD$ taking $\bx_1, ..., \bx_n$ to be i.i.d. $\bern(\frac{1}{2})$ random variables. Similarly, we can sample from $\calD'$ by taking $\bx_1, ..., \bx_{n-1}$ to be i.i.d. $\bern(\frac{1}{2})$ and then setting $\bx_n = 1 \oplus \bigoplus_{i=1}^{n-1} x_i$. We now compute
	\begin{align*}
		& \hspace{-0.4cm}\E_{\bx \sim \calD'} \left [ \dist(\bx, \{0^n, 1^n\}) \right]\\ &= \E_{\bx \sim \calD'} \left [ \dist(\bx_1\bx_2...\bx_{n-1}, \{0^{n-1}, 1^{n-1}\}) + \mathbb{I}_{\bx_n \not = \maj(\bx_1, ..., \bx_{n-1})} \right] \\
		& 
	    	= \E_{\bx \sim \calD', \by \sim \bern(\frac{1}{2})} [\dist(\bx_1\bx_2...\bx_{n-1}, \{0^{n-1}, 1^{n-1}\}) + \mathbb{I}_{\bx_n \not = \maj(\bx_1, ..., \bx_{n-1})} \\
	    	& \hspace{2cm}+ \mathbb{I}_{\by \not = \maj(\bx_1, ..., \bx_{n-1})} - \mathbb{I}_{\by \not = \maj(\bx_1, ..., \bx_{n-1})} ]
	    \\
	  	&= \E_{\bx \sim \calD} [\dist(\bx, \{0^n, 1^n\})] + \E_{\bx \sim \calD', \by \sim \bern(\frac{1}{2})} [ \mathbb{I}_{\bx_n \not = \maj(\bx_1, ..., \bx_{n-1})} - \mathbb{I}_{\by \not = \maj(\bx_1, ..., \bx_{n-1})}] \\
	  	&= \E_{\bx \sim \calD} [\dist(\bx, \{0^n, 1^n\})] + \E_{\bx \sim \calD'} [ \mathbb{I}_{\bx_n \not = \maj(\bx_1, ..., \bx_{n-1})}] - \frac{1}{2}
	\end{align*}
	where the first equality uses the fact the fact that $\maj(\bx_1, ..., \bx_n) = \maj(\bx_1, ..., \bx_{n-1})$ since $n$ is a multiple of $4$ and $x$ has odd parity. Now note that 
		\[\Pr_{\bx \sim \calD'} \left[ \bx_n \not = \maj(\bx_1, ..., \bx_{n-1}) \right] = \E_{\bi \sim [n]} \left[ \Pr_{\bx \sim \calD'} [\bx_{\bi} \not = \maj(\bx_{-\bi})] \right] = \E_{\bx \sim \calD'} \left[ \Pr_{\bi \sim [n]} [\bx_{\bi} \not = \maj(\bx_{-\bi})] \right].\]
	Since $4|n$ we have that any string $x$ of odd parity satisfies $\maj(x_{-j}) = \maj(x)$ for all $j$. Moreover, at most $n/2 - 1$ bits of $x$ differ from the majority. Thus,
	\[\E_{\bx \sim \calD'} \left[ \Pr_{\bi \sim [n]} [\bx_{\bi} \not = \maj(\bx_{-\bi})] \right] \leq \left (\frac{n}{2} - 1 \right) \frac{1}{n} = \frac{1}{2} - \frac{1}{n}\]
	as desired.
	
\end{proof}

%% file: Sections/TestingGaps.tex

\section{Lower Bounds for Relaxed Junta Testing}

Recall that for relaxed tolerant junta testing, a $(k, k', \eps_1, \eps_2)$ tester must accept functions that are $\eps_1$-close to some $k$-junta and reject those that are $\eps_2$-far from all $k' > k$ juntas. To prove Theorem \ref{thm:tolerant-gap-lower-bound}, it will suffice to show:

\begin{theorem}
\label{thm:weak-gap-tolerant-lower-bound}
For any $0.01 \leq \eps_1 < \eps_2 \leq 0.49$ and integer $k$ with $\eps_2 - \eps_1 \geq 2^{-O(k)}$ and 
any $(k, k + \lfloor \frac{-\log(\eps_2 - \eps_1)}{20} \rfloor, \eps_1, \eps_2)$-tolerant junta tester must make at least $\left (\frac{k}{-\log(\eps_2 - \eps_1)} \right)^{-\Omega(\log(\eps_2 - \eps_1))}$ queries.
\end{theorem}

We can then combine this result with the following observation:

\begin{lemma}
\label{lem:lifting}
Let $f: \{0,1\}^n \rightarrow \{0,1\}$ be a boolean function, $b$ be an integer, and take $F: \{0,1\}^{nb} \rightarrow \{0,1\}$ as
	\[F(x) = f \left( \bigoplus_{i=1, \hdots, b} x_i, \bigoplus_{i = b+1, \hdots, 2b} x_i, \hdots , \bigoplus_{i = nb-b+1, \hdots, nb} x_i \right).\]
Then for any $k \leq n$ we have that 
	\[\dist(F, \calJ_{kb}) = \dist(F, \calJ_{kb + b - 1}) = \dist(f, \calJ_k)\]
\end{lemma}

We include a proof in the appendix, but intuitively any junta must use all the coordinates from a set of xor'd variables or no variables from it. This then gives the following corollary:

\begin{corollary}
\label{cor:lower-bound-lift}
	Let $k, \ell, b$ be integers and $0 \leq \eps_1 \leq \eps_2 \leq \frac{1}{2}$. If any $(k, k + \ell, \eps_1, \eps_2)$ tester must  make $Q(k, \ell, \eps_1, \eps_2)$ queries, then any $(kb, (k + \ell)b + b - 1, \eps_1, \eps_2)$ tester must make $Q(k, \ell, \eps_1, \eps_2)$ queries.
\end{corollary}

\begin{proof}
We'll construct a $(k, k + \ell, \eps_1, \eps_2)$ tester using a $(kb, (k + \ell) b + b - 1, \eps_1, \eps_2)$ tester $\ALG$. Indeed, given a function $f: \{0,1\}^n \rightarrow \{0,1\}$, we construct $F$ as in Lemma \ref{lem:lifting}. We then run $\ALG$ on $F$ and accept if $\ALG$ accepts and reject otherwise. 

Note that $\dist(F, \calJ_{kb}) = \dist(f, \calJ_k)$ and $\dist(F, \calJ_{(k + \ell) b + b -1}) = \dist(f, \calJ_{k + \ell})$, so this indeed constitutes a $(k, k + \ell, \eps_1, \eps_2)$ tester. Since we can answer each query to $F$ with at most one query to $f$, the corollary follows.
\end{proof}

With this, we can prove Theorem \ref{thm:tolerant-gap-lower-bound}.

\begin{proof}[Proof of Theorem \ref{thm:tolerant-gap-lower-bound}]
Note that it suffices to handle the case when $\gamma \geq \frac{1}{k^{.01}}$: If $\frac{1}{k} \leq \gamma \leq \frac{1}{k^{.01}}$, a $(k, (1 + \gamma)k, \eps_1, \eps_2)$ tester is also a $(k, (1 + \frac{1}{k^{.01}})k, \eps_1, \eps_2)$ tester. Thus, the $k^{-\Omega(\log(\eps_2 - \eps_1))}$ query lower bounds for a $(k, (1 + \frac{1}{k^{.01}})k, \eps_1, \eps_2)$ applies to the $(k, (1 + \gamma)k, \eps_1, \eps_2)$ tester, which proves the desired result.

Let $\ell = \lfloor -\log(\eps_2 - \eps_1) \rfloor$ and set $k' = \lfloor \ell/(100 \gamma) \rfloor$. Now observe that if $\gamma$ is smaller than some appropriate absolute constant, we can apply Theorem \ref{thm:weak-gap-tolerant-lower-bound} to get a $\left( \frac{1}{\gamma} \right)^{- \Omega(\log(\eps_2 - \eps_1))}$ lower bound for a $(k', k' + \lfloor \ell /20 \rfloor, \eps_1, \eps_2)$ tester. Now let $b = \lfloor k/k' \rfloor$. Applying Corollary \ref{cor:lower-bound-lift} then extends this lower bound to $(k'b, k'b + b \lfloor \ell /20 \rfloor + b - 1, \eps_1, \eps_2)$ testers. 

Note note that $k' \leq k^{.12}$. Thus
	\[k'b + b \lfloor \ell /20 \rfloor + b - 1\geq k - k' + 5 k \gamma - 5 k' \gamma - 1 \geq (1 + 4 \gamma)k > (1 + \gamma)k \]
assuming $k$ is sufficiently large. 
\end{proof}

\subsection{Weak Gap Lower Bound}

It now remains to prove Theorem \ref{thm:weak-gap-tolerant-lower-bound}. At a high level, we follow the same proof as with Theorem \ref{thm:tolerant-lower-bound} but with some changes the $\calD_{NO}$ distribution. Since most of the proofs are simple or identical to their counterparts in the original lower bound we banish them to the appendix.

\begin{lemma}
\label{lem:radial-coloring}
	Let $d \leq n$ be integers. There exists a coloring of the boolean cube $\chi: \{0,1\}^n \rightarrow [|B_n(d)|]$ such that for all $x,y \in \{0,1\}^n$ with $\dist(x,y) \leq d$, $\chi(x) \not = \chi(y)$.
\end{lemma}

We'll also need the following fact.

\begin{lemma}
\label{lem:random-bucket-expectation}
Let $\lambda_1, ..., \lambda_n$ be non-negative real numbers with $\sum_i \lambda_i = 1$. Let $\bX_1, ..., \bX_n$ be drawn uniformly and independently at random from $\{-1,1\}$. Then
	\[\E \left [ \left |\sum_i \lambda_i \bX_i \right | \right] \geq \E \left [ \left |\sum_i \frac{1}{n} \bX_i \right | \right]. \]
\end{lemma}

Finally, standard results about random walks give us that

\begin{lemma}[Folklore]
\label{lem:delta_properties}
$\Delta_t$ satisfies the following properties
\begin{enumerate}
\item[(i)] $\Delta_t$ is an increasing function in $t$
\item[(ii)] For $t$ sufficiently large, $\frac{1}{2} - \frac{10}{\sqrt{t}} \leq \Delta_t \leq \frac{1}{2} - \frac{1}{10 \sqrt{t}}$
\end{enumerate}
\end{lemma}

With this we have everything we need to construct our new distributions. We again proceed by Yao's lemma and construct a $\calD_{YES}$ and $\calD_{NO}$ distributions, which are again boolean functions over $n := k + \ell$ variables. We take $\ell, r$ and $d$ to be parameters, which we'll set later. We again will ensure that $r < k$ and take $d = \Theta(\ell)$. 

\begin{flushleft}
\begin{enumerate}
    	\item[] $\calD_{NO}$: 
	A boolean function $\bff:\{0,1\}^n\rightarrow\{0,1\}$ is drawn as follows:
	\begin{enumerate}
	\item For each $x\in \{0,1\}^n$ with $x|_{[r]} \not = 0^r$, independently draw $\bff(x)$ from $\bern(p)$.
	\item For each $x\in \{0,1\}^n$ with $x|_{[r]} = 0^r$, draw $\bff(x)$ uniformly
	and independently at random.
    \end{enumerate}
	    \item[] $\calD_{YES}$: A boolean function $\bff:\{0,1\}^n\rightarrow \{0,1\}$ is drawn as follows:
	\begin{enumerate}
	    \item Randomly choose a set $\bJ \subseteq \{r+1, ..., n\}$ of size $\ell$.
	    \item Let $\chi$ be a coloring of $\{0,1\}^{\bJ}$ from Lemma \ref{lem:radial-coloring} such that points within distance $d$ have different colors.
	   	\item For each $\rho \in \{0,1\}^{[n] \setminus \bJ}$, sample $\by^\rho_1, ..., \by_{|B_\ell(d)|}^\rho \sim \bern(\frac{1}{2})$.
	    \item For each $x\in \{0,1\}^n$ with $x|_{[r]} \not = 0^r$, independently draw $\bff(x)$ from $\bern(p)$. 
	    \item For each $x\in \{0,1\}^n$ with $x|_{[r]} = 0^r$ set $\bff(x) = \by^{x|_{[n] \setminus \bJ}}_{\chi(x|_{\bJ})}$.  	
	\end{enumerate}
\end{enumerate}
\end{flushleft}	

We now follow the previous proof. Note that by construction we have that

\begin{lemma}
\label{lem:weak-gap-almost-same-distribution}
Consider a set of points $x^{(1)}, ..., x^{(m)} \in \{0,1\}^n$ and $y_1,\ldots,y_m\in \{0,1\}$.
Then
	\[\Pr_{\bff, \bJ \sim \calD_{YES}} \left[\forall i\hspace{2mm}  \bff(x^{(i)}) = y_i \mid \forall i,j  \quad x^{(i)}|_{[n] \setminus \bJ} \not = x^{(j)}|_{[n] \setminus \bJ} \lor \dist(x,y) \leq d \right] = \Pr_{\bff \sim \calD_{NO}} \left[\forall i \hspace{2mm} \bff(x^{(i)}) = y_i \right]\]
\end{lemma}

\begin{lemma}
\label{lem:weak-gap-query-lower-bound}
	Any deterministic algorithm $\ALG$ that distinguishes between $\calD_{YES}$ and $\calD_{NO}$ with probability at least $2/3$, i.e., $\Pr_{\bff \sim \calD_{YES}} [\ALG(\bff) \text{ accepts}] \geq 2/3$ and $\Pr_{\bff \sim \calD_{NO}} [\ALG(\bff) \text{ rejects}] \geq 2/3$, must make at least $\Omega((k/\ell)^{d/2})$ queries. 
\end{lemma}

Applying Lemma \ref{lem:no-distribution-distance} with our new parameters then gives

\begin{lemma}
\label{lem:weak-gap-no-distribution-distance}
Suppose that $\ell \leq k$, then 	
	\[\Pr_{\bff \sim \calD_{NO}} \left[ \dist(\bff, \calJ_{k + \ell/10}) \leq  p (1-2^{-r}) + \Delta_{2^{0.9\ell}} \cdot 2^{-r} - e^{-2^{0.9\ell} \cdot (1/2 - p)^2/12} - 2^{-(k + \ell/10)/3} \right] = o_k(1)\]
where $\calJ_{k + \ell/10}$ denotes the class of $(k + \ell/10)$-juntas.
\end{lemma}

Finally, we have that
\begin{lemma}
\label{lem:weak-gap-d-yes-distance} 	
Suppose that $\ell \geq 2$, then
	\[\Pr_{\bff \sim \calD_{YES}} \left[ \dist(\bff, \calJ_{k}) \geq p (1-2^{-r}) + \Delta_{|B_\ell(d)|} \cdot 2^{-r} + e^{-2^\ell \cdot (1/2 - p)^2/12} + 2^{-k/3} \right] = o_k(1)\]
where $\calJ_{k}$ denotes the class of $k$-juntas.
\end{lemma}

Combining these all together and setting parameters now gives us the theorem:

\begin{proof}[Proof of Theorem \ref{thm:weak-gap-tolerant-lower-bound}]
	It only remains to set the parameters so any $(k, \eps_1, \eps_2)$ tolerant tester is forced to distinguish the two distributions. We let $\ell = 10 \lfloor -\log(1-\eps_1 \eps_2^{-1}) \rfloor$. We then take $r$ to be the smallest integer such that
		\[\Delta_{2^{0.9\ell}} 2^{-r} \leq \eps_2 + 2^{-(k + 0.1 \ell)/3} + e^{-2^{0.9\ell} \cdot (0.009)^2/12} \]
	and set $p$ to satisfy
		\[ \eps_2 + 2^{-(k + 0.1 \ell)/3} + e^{-2^{0.9\ell} \cdot (0.009)^2/12} = p (1-2^{-r}) + \Delta_{2^{0.9\ell}} \cdot 2^{-r} .\]
	We assume that $\ell$ and $k$ are sufficiently large such that $2^{-(k + 0.1 \ell)/3} \leq .005$, $e^{-2^{0.9 \ell} \cdot (0.009)^2/12} \leq .005$, and $\Delta_{2^{0.9\ell}} \geq .491$. Note that these together imply that $r \geq 0$ and $0 \leq p \leq .491$. Moreover, note that by the minimality of $r$, we have that
		\[\Delta_{2^{0.9\ell}} 2^{-r + 1} \geq \eps_2.\] 
	which implies that $2^{-r} \geq \eps_2$, yielding $r = O(1)$ and thus that $r < k$ as promised.
	
	Finally, we set $d$ to be the largest integer such that
	 	\[|B_{\ell}(d)| \leq 2^{0.1 \ell}.\]
	This clearly implies that $d = \Theta(\ell)$ as promised.
	
	 To see that a $(k, \eps_1, \eps_2)$ tolerant tester must distinguish $\calD_{YES}$ and $\calD_{NO}$, first observe that by Lemma \ref{lem:weak-gap-no-distribution-distance} and our choice of parameters, we have that functions in $\calD_{NO}$ are at least $\eps_2$-far from the set of $(k + \ell/10)$-juntas with high probability. On the other hand, by Lemma \ref{lem:weak-gap-d-yes-distance} and Lemma 
	 \ref{lem:delta_properties}, functions from $\calD_{YES}$ are with high probability at most
	 	\[p (1-2^{-r}) + \Delta_{2^{0.1\ell}} \cdot 2^{-r} + e^{-2^\ell \cdot (.009)^2/12} + 2^{-k/3}\]
	 close to some $k$-junta. We'll show this is at most $\eps_1$ under our choice of parameters. Note this is at equal to
	 	\[\eps_2 + \Delta_{2^{0.1\ell}} \cdot 2^{-r} - \Delta_{2^{0.9\ell}} \cdot 2^{-r} + e^{-2^\ell \cdot (.009)^2/12} + e^{-2^{0.9\ell} \cdot (0.009)^2/12} + 2^{-k/3} + 2^{-(k + 0.1 \ell)/3}\]
	Again by Lemma \ref{lem:delta_properties}, we have that if $\eps_2 - \eps_1$ is sufficiently small
	 	\[\Delta_{2^{0.1\ell}} \cdot 2^{-r} - \Delta_{2^{0.9\ell}} \cdot 2^{-r} \leq 10 \sqrt{2^{-0.9\ell}} - \frac{1}{10} \sqrt{2^{-0.1\ell}} \leq - \frac{1}{20} \sqrt{2^{-0.1\ell}} \]
	Combining this with our hypothesis that $\eps_2 - \eps_1 \geq 2^{k/6}$ and the fact that $2^{-r} \geq \eps_2$, we get that the distance of functions from $\calD_{YES}$ is at most 
	 	\[\eps_2 - \frac{1}{20} \sqrt{2^{-0.1\ell}} \cdot \eps_2 + 4(\eps_1 - \eps_2)^2 \leq \eps_2 - \frac{1}{20} (1 - \eps_1 \eps_2^{-1})^{\frac{1}{2}} \eps_2 + (\eps_2 - \eps_1)^2 \leq \eps_2 - 2 (1 - \eps_1 \eps_2^{-1}) \eps_2 + (\eps_2 - \eps_1)^2 < \eps_1\]
	 again assuming $\eps_2 - \eps_1$ is sufficiently small. The proof is now complete since	
	 	\[\ell/10 \geq -\log(\eps_2 - \eps_1) + \log(\eps_2) - 1 \geq -\log(\eps_2 - \eps_1)/20.\]
	since $\eps_2 \geq 0.01$. 
\end{proof}

%% file: Sections/Barriers.tex

\section{Barriers to Stronger Lower Bounds}

There are several interesting open questions raised by this work.
Can we further improve the lower bound for tolerant junta testing? 
For the relaxed model, can we rule out $(k,2k,\eps_1,\eps_2)$-testers
  that make $\text{poly}(k,(\eps_2-\eps_1)^{-1})$ queries? Additionally, can we use our techniques to prove lower bounds for testing monotonicity (tolerantly)? 

We conclude with some limitations of our approach to getting stronger bounds. Optimistically, one may hope that asking for smaller than $(2^\ell - 1)$-wise independence may give better bounds, but this will not work naively even if we change $\calD_{NO}$. Indeed, suppose we have a distribution $\calD_{NO}$ over boolean functions on $\{0,1\}^n$ with $n := k + \ell$. Let $\calD_{YES}^J$ for a subset $J \subseteq [n]$ of size $\ell$ also be a distribution over boolean functions on $\{0,1\}^n$. Moreover, assume that when $\bff \sim \calD_{YES}^J$ and $\bgg \sim \calD_{NO}$ we have that for any $\rho \in \{0,1\}^{[n] \setminus J}$ and any $m$ points $y^{(1)}, ..., y^{(m)} \in \{0,1\}^J$, $\bff(\rho \sqcup y^{(1)}), ..., \bff(\rho \sqcup y^{(m)})$ and $\bgg(\rho \sqcup y^{(1)}), ..., \bgg(\rho \sqcup y^{(m)})$ are identically distributed. 

It turns out that under these assumptions, if we follow our proof strategy naively we cannot prove better lower bounds. To see this, we will need the following lemma.

\begin{lemma}
\label{lem:k-wise-limit}
Let $\calD_1$ and $\calD_2$ be distributions over boolean functions $f: \{0,1\}^\ell \rightarrow \{0,1\}$. Moreover, suppose that for any $\bx^{(1)}, ..., \bx^{(m)} \in \{0,1\}^\ell$ and $y_1, .., y_m \in \{0,1\}$ we have that
	\[\Pr_{\bff \sim \calD_1} \left[\forall i \quad \bff(x_i) = y_i \right] = \Pr_{\bgg \sim \calD_2} \left[\forall i \quad  \bgg(x_i) = y_i \right] \]
Then 
	\[\bigg| \E_{\bff \sim \calD_1} \left [ \dist(\bff, \{0,1\}) \right] - \E_{\bgg \sim \calD_2} \left[ \dist(\bgg, \{0,1\}) \right]  \bigg| \leq O \left(\sqrt{\frac{\log(m)}{m}} \right). \]
where $\dist(f,\{0,1\}) := \min \{ \Pr_{\bx \sim \{0,1\}^\ell}[f(x) = 0], \Pr_{\bx \sim \{0,1\}^\ell}[f(x) = 1]\}$.
\end{lemma}

We leave the proof for the appendix. That said, intuitively such a pair of distributions would give a lower bound against estimating the distance of a function $f: \{0,1\}^\ell \rightarrow \{0,1\}$ to constant, and we know that this can be done without many samples.

Now note that an algorithm that makes $n^r$ queries can query every point in a ball $B(x,r)$ for some $x \in \{0,1\}^n$. Thus, in order to avoid giving any evidence that could distinguish $\calD_{YES}^J$ and $\calD_{NO}$, we must take $m \geq \binom{\ell}{r}$. But Lemma \ref{lem:k-wise-limit}, implies that if $\bff$ and $\bgg$ are $\epsilon_1$-close and $\epsilon_2$-far from being a junta on $[n] \setminus J$ with high probability then $m \leq O((\epsilon_2 - \epsilon_1)^{-3})$. Simplifying, we see that we at best get a $n^{-\Omega(\log(\epsilon_2 - \epsilon_1))}$ lower bound.

Moreover, we can also observe that a better lower bound cannot use a uniformly random functions as $\calD_{NO}$: Namely, we can distinguish a uniformly random distribution from a distribution on functions closer than random to $k$-juntas by sampling a random $x \in \{0,1\}^n$ and querying all points within hamming distance $r$. We then check if there is a set $S \subseteq [n]$ of size $k$ such that $B(x,r) \cap \{y \in \{0,1\}^n: y_S = x_S\}$ is biased. For $r$ suitably large, we expect few biased balls under the random distribution, but for functions that are closer than random to juntas we expect to see a biased ball with reasonable probability. Formally,

\begin{lemma}
\label{lem:testing-random-functions}
Let $\calD_{YES}$ a distribution over boolean functions $f:\{0,1\}^n \rightarrow \{0,1\}^n$ such that 
		\[\Pr_{\bff \sim \calD_{YES}}[\dist(\bff, \calJ_k)\geq \frac{1}{2} - \eps] = o(1)\]
where $\eps \geq \Omega(2^{-(n-k)/10})$ and $\eps \geq 2^{-n/128}$. Take $\calD_{NO}$ to be the uniform distribution over all boolean functions. Then there exists an algorithm $\ALG$ that makes at most $n^{O(\log(1/\eps))}$ queries and distinguishes $\calD_{YES}$ and $\calD_{NO}$ with probability $2/3$.
\end{lemma}

We again leave the proof for the appendix. Note that this rules out the possibility of a better lower bound with $n = \poly(k)$. When $n = k^{\omega(1)}$, we remark that the techniques of \cite{iyer2021junta} are likely able to remove our testers dependence on $n$ and give a $(k/\eps^2)^{O(\log(1/\eps))}$ query bound.

Thus, further improvements must move beyond $m$-wise independence and random functions. A promising way of circumventing both these barriers would be to consider using distributions $\calD_1$ and $\calD_2$ over functions $f:\{0,1\}^\ell \rightarrow \{0,1\}$ that are only identically distributed on balls $B(x,r)$ for $x \in \{0,1\}^\ell$ rather than distributed identically for any $m$ points. Indeed, \cite{pallavoor2022approximating} shows that there are distributions $\calD_1$ and $\calD_2$ over functions $f: \{0,1\}^n \rightarrow \{0,1\}$ with functions in $\calD_2$ being $\Omega(1)$ closer (in expectation) to constant than those from $\calD_1$ and such that $\calD_1$ and $\calD_2$ are identically distributed along any ball $B(x,r)$ of radius $O(\sqrt{n})$.

%% file: Sections/AppendixA.tex
\appendix
\section{Missing Proofs from Section 3}

\begin{proof}[Proof of Lemma \ref{lem:lifting}] 
	For simplicity of notation, we'll assume $f: \{ \pm 1\}^n \rightarrow \{ \pm 1\}$. We'll show that $\dist(F,\calJ_{kb}) \leq \dist(f, \calJ_k)$ and $\dist(F, \calJ_{kb + b -1}) \geq \dist(f, \calJ_k)$. Since $\dist(F, \calJ_{kb}) \geq \dist(F, \calJ_{kb + b - 1})$ the result follows.
	
	We start by showing $\dist(F,\calJ_{kb}) \leq \dist(f, \calJ_k)$. Indeed, let $S \subseteq \{\pm 1\}^n$ be a minimum set of changes needed to make $f$ into a $k$-junta. Let
		\[ m(x) = \left ( \bigoplus_{i = 1 ... b} x_i, \bigoplus_{i = b + 1, ..., 2b} x_i, \hdots, \bigoplus_{i = nb - b + 1, ..., nb} x_i \right)\]
	To get a $kb$-junta, it clearly suffices to change the values of $F$ under $m^{-1}(S)$. Since $m^{-1}(S) = |S| 2^{nb}/ 2^n$, we conclude $\dist(F,\calJ_{kb}) \leq \dist(f, \calJ_k)$.

	We now claim that $\dist(F, \calJ_{kb + b -1}) \geq \dist(f,\calJ_k)$. Indeed, fix a set $J \subseteq [nb]$ of size $kb + b - 1$. Without loss of generality, we let $\ell$ be the largest integer such that $[0, \ell b] \subseteq J$ and assume every subsequent block is missing at least one element. Note that fixing a prefix $\rho \in \{\pm 1\}^J$ will then fix the first $\ell$ bits of $m(x)$ to $z^\rho \in \{\pm 1\}^\ell$. We can then compute
		\begin{align*}
			\dist(g, \calJ_J) &= \frac{1}{2^{kb + b - 1}} \sum_{\rho \in \{\pm 1\}^J} \frac{1}{2} - \frac{1}{2^{n b - kb - b + 2} }\left| \sum_{y \in \{\pm 1\}^{[n b] \setminus J}} g(\rho \sqcup y) \right|	 \\
			&= \frac{1}{2^{kb + b - 1}} \sum_{\rho \in \{\pm 1\}^J} \frac{1}{2} - \frac{1}{2^{n b - kb - b + 2} } \cdot \frac{2^{nb-kb - b + \ell + 1}}{2^n} \left| \sum_{b \in \{\pm 1\}^{[n] \setminus [\ell]}} f(z^\rho \sqcup b) \right| \\
			&= \frac{1}{2^{kb + b - 1}} \frac{2^{kb+b - 1}}{2^\ell} \sum_{a \in \{\pm 1\}^\ell}  \frac{1}{2} - \frac{2^{\ell}}{2^{n+1}} \left| \sum_{b \in \{\pm 1\}^{[n] \setminus [\ell]}} f(a \sqcup b) \right| \\
			&= \frac{1}{2^\ell} \sum_{a \in \{\pm 1\}^\ell}  \frac{1}{2} - \frac{1}{2^{n - \ell +1}} \left| \sum_{b \in \{\pm 1\}^{[n] \setminus [\ell]}} f(a \sqcup b) \right| = \dist(f, \calJ_{[\ell]}) \\			\end{align*}
	
	The result now follows since $\ell$ is at most $k$.
\end{proof}

\begin{proof}[Proof of Lemma \ref{lem:radial-coloring}]
If we make a graph $G$ with vertices $\{0,1\}^n$ and an edge between $x,y$ with $\dist(x,y) \leq d$. Note every vertex in $G$ has degree at most $B_n(d) - 1$. It then follows by a greedy coloring that we need at most $|B_n(d)|$ colors.
\end{proof}

\begin{proof}[Proof of Lemma \ref{lem:random-bucket-expectation}]
	Without loss of generality suppose that $\lambda_1 > \lambda_2$. We claim that 
		\[\E \left [ \left |\sum_i \lambda_i \bX_i \right | \right] \geq \E \left [ \left |\frac{\lambda_1 + \lambda_2}{2} (\bX_1 + \bX_2) + \sum_{i = 3}^n \bX_i \right | \right]. \]
	Let $\overline{\lambda} = \frac{\lambda_1 + \lambda_2}{2}$. Now note
		\begin{align*}
			 \E & \left [ \left |\sum_{i=1}^n \lambda_i \bX_i \right | \right] - \E \left [ \left |\overline{\lambda} (\bX_1 + \bX_2) + \sum_{i=3}^n \lambda_i \bX_i \right | \right] \\
			 &= \E_{\bX_3, \hdots  \bX_n} \left [ \E_{\bX_1, \bX_2} \left[ \left |\sum_{i=1}^n \lambda_i \bX_i \right | - \left| \overline{\lambda} (\bX_1 + \bX_2) + \sum_{i = 3}^n \lambda_i \bX_i \right| \bigg| \bX_3, ..., \bX_n \right ] \right ] \\
			&=  \frac{1}{4} \E_{\bX_3, \hdots  \bX_n} \left [ \left |\lambda_1 - \lambda_2 + \sum_{i=3}^n \lambda_i \bX_i \right | + \left |\lambda_2 - \lambda_1 + \sum_{i=3}^n \lambda_i \bX_i \right | - 2 \left| \sum_{i = 3}^n \lambda_i \bX_i \right| \right ] \\
		\end{align*}
	which is non-negative by Jensen's inequality. The lemma now follows by a limiting argument.
\end{proof}

\begin{proof}[Proof of Lemma \ref{lem:weak-gap-query-lower-bound}]
Towards a contradiction, suppose $\ALG$ distinguishes the distributions and makes at most $Q = \frac{1}{10} (k/\ell)^{d/2}$ queries. Call a pair of points $x,y \in \{0,1\}^n$ bad if $x|_{[n] \setminus \bJ} = y|_{[n] \setminus \bJ}$ and $\dist(x,y) \geq d$. Since $\ALG$ is deterministic, it corresponds to a decision tree. But now observe that for a particular path $p$ of the decision tree, we have that
			\[\Pr_{\bff \sim \calD_{NO}}[\ALG(\bff) \text{ follows } p] = \Pr_{\bff, \bJ \sim \calD_{YES}}[\ALG(\bff) \text{ follows } p |  p \text{ doesn't query a bad pair} ]. \]
		by Lemma \ref{lem:weak-gap-almost-same-distribution}. We then conclude that
		\[\Pr_{\bff \sim \calD_{NO}}[\ALG(\bff) \text{ follows } p] \leq \frac{\Pr_{\bff \sim \calD_{YES}}[\ALG(\bff) \text{ follows } p]}{\Pr_{\bff, \bJ \sim \calD_{YES}} [p \text{ doesn't query a bad pair}]}. \]
		Now observe		
			\[ \Pr_{\bff, \bJ \sim \calD_{YES}} [p \text{ doesn't query a bad pair}] \geq 1 - Q^2 \cdot (\ell/k)^d \geq .99. \]
		Thus
			\[ \Pr_{\bff \sim \calD_{NO}}[\ALG(\bff) \text{ follows } p] \leq 1.02 \cdot \Pr_{\bff \sim \calD_{YES}}[\ALG(\bff) \text{ follows } p] \]
		 Summing over all rejecting paths we conclude that
			\[\Pr_{\bff \sim \calD_{NO}}[\ALG(\bff) \text{ rejects} ] \leq 1.02 \cdot \Pr_{\bff \sim \calD_{YES}}[\ALG(\bff) \text{ rejects}] < 1/2, \]
		a contradiction. So any tester must make $\Omega((k/\ell)^{d/2})$ queries as claimed.	
\end{proof}

\begin{proof}[Proof of Lemma \ref{lem:weak-gap-d-yes-distance}]
	We claim that $\bff$ is close to a junta on $[n] \setminus \bJ$:
	\begin{align*}
		&\hspace{-0.4cm} \E_{\bff, \bJ \sim \calD_{YES}}[\dist(\bff, \calJ_{[n] \setminus \bJ})]\\ &= \frac{1}{2^k} \E_{\bff, \bJ \sim \calD_{YES}} \left[ \sum_{\rho \in \{0,1\}^{[n] \setminus \bJ}}  \sum_{y \in \{0,1\}^{\bJ}} \frac{1}{2^\ell} \left|\bff(\rho \sqcup y) - \maj(\{\bff(\rho \sqcup y): y \in \{0,1\}^{\bJ} \}) \right| \right]	\\
		& \leq p (1-2^{-r}) + e^{-2^\ell \cdot (1/2 - p)^2/12} + \frac{1}{2^k} \E_{\bff, \bJ \sim \calD_{YES}} \left[ \sum_{\substack{\rho \in \{0,1\}^{[n] \setminus \bJ} \\ \rho|_{[r]} = 0^r}} \sum_{y \in \{0,1\}^{\bJ}} \frac{1}{2^\ell} \left|\bff(\rho \sqcup y) - \maj(\{\bff(\rho \sqcup y): y \in \{0,1\}^{\bJ} \}) \right| \right]
	\end{align*}
	by Lemma \ref{lem:unbalanced-distance-to-constant}. To bound the second term, let $\lambda_i = |\chi^{-1}(i)|/2^\ell$ and observe that for a fixed $\rho$ with $\rho|_{[r]} = 0^r$ we have that
	\begin{align*}
	\E_{\bff, \bJ \sim \calD_{YES}} &\left[ \sum_{y \in \{0,1\}^{\bJ}} \frac{1}{2^\ell} \left|\bff(\rho \sqcup y) - \maj(\{\bff(\rho \sqcup y): y \in \{0,1\}^{\bJ} \}) \right| \right] \\
		&= \frac{1}{2} - \frac{1}{2} \E_{\bX_i \sim \{-1,1\}} \left[ \left| \sum_{i = 1}^{|B_\ell(d)|} \lambda_i \bX_i \right | \right] \leq \frac{1}{2} - \frac{1}{2} \E_{\bX_i \sim \{-1,1\}} \left[ \left| \sum_{i = 1}^{|B_\ell(d)|} \frac{1}{|B_\ell(d)|} \bX_i \right| \right] = \Delta_{|B_\ell(d)|}
	\end{align*}
	by Lemma \ref{lem:random-bucket-expectation}. Thus,
		\[\E_{\bff, \bJ \sim \calD_{YES}}[\dist(\bff, \calJ_{[n] \setminus \bJ})] \leq p (1-2^{-r}) + e^{-2^\ell \cdot (1/2 - p)^2/12} + \Delta_{|B_\ell(d)|} 2^{-r} \]
	as desired.
\end{proof}

%% file: Sections/AppendixB.tex
\section{Missing Proofs from Section 4}
\begin{proof}[Proof of Lemma \ref{lem:k-wise-limit}]
	Fix some function $h: \{0,1\}^\ell \rightarrow \{\pm 1\}$ and choose $\bx^{(1)}, ..., \bx^{(m)} \in \{0,1\}^\ell$ uniformly and independently at random. Note that
		\[ \E_{\bx^{(i)}} [h(\bx^{(i)})] = \frac{\sum_{x \in \{0,1\}^\ell} h(x) }{2^\ell}. \]
	Moreover, since $h(\bx^{(1)}), ..., h(\bx^{(m)})$ are independent Chernoff-Hoeffding bounds give
	
		\[\Pr_{\bx^{(1)}, ..., \bx^{(m)}} \left[ \left| \frac{1}{m} \sum_{i=1}^m h(\bx^{(i)}) - \frac{\sum_{x \in \{0,1\}^\ell} h(x) }{2^\ell} \right| \geq \eps \right] \leq 2e^{- \Omega(\eps^2 m)} \]
	So we conclude that 
		\[\left | \E_{\bx^{(1)}, ..., \bx^{(m)}} \left [ \left| \frac{1}{m} \sum_{i=1}^m h(\bx^{(i)}) \right | - \frac{\left | \sum_{x \in \{0,1\}^\ell} h(x) \right |}{2^\ell}  \right ] \right| \leq \eps + 4e^{- \Omega(\eps^2 m)}. \]
	It now follows that
		\[\E_{\bff \sim \calD_1} \left[ \left | \E_{\bx^{(1)}, ..., \bx^{(m)}} \left [ \left| \frac{1}{m} \sum_{i=1}^m (2\bff(\bx^{(i)}) -1 ) \right | - \frac{\left | \sum_{x \in \{0,1\}^\ell} (2\bff(x) - 1) \right |}{2^\ell}  \right ] \right| \right] \leq \eps + 4e^{- \Omega(\eps^2 m)}. \]
	By Jensen's inequality,
		\[\left | \E_{\bff \sim \calD_1, \bx^{(1)}, ..., \bx^{(m)}} \left[  \left| \frac{1}{m} \sum_{i=1}^m (2\bff(\bx^{(i)})-1) \right | \right] - \E_{\bff \sim \calD_1} \left[ \frac{\left | \sum_{x \in \{0,1\}^\ell} (2\bff(x)-1) \right |}{2^\ell}  \right] \right| \leq \eps + 4e^{- \Omega(\eps^2 m)}. \]		

	Analogously, we have that
		\[\left | \E_{\bgg \sim \calD_2, \bx^{(1)}, ..., \bx^{(m)}} \left[  \left| \frac{1}{m} \sum_{i=1}^m (2\bgg(\bx^{(i)})-1) \right | \right] - \E_{\bgg \sim \calD_2} \left[ \frac{\left | \sum_{x \in \{0,1\}^\ell} (2\bgg(x)-1) \right |}{2^\ell}  \right] \right| \leq \eps + 4e^{- \Omega(\eps^2 m)}. \]		
	But now as 
		\[\E_{\bff \sim \calD_1, \bx^{(1)}, ..., \bx^{(m)}} \left[  \left| \frac{1}{m} \sum_i (2\bff(\bx^{(i)})-1) \right | \right] = \E_{\bgg \sim \calD_2, \bx^{(1)}, ..., \bx^{(m)}} \left[  \left| \frac{1}{m} \sum_i (2\bgg(\bx^{(i)})-1) \right | \right]\]
	we conclude that
	\[\left| \E_{\bff \sim \calD_1} \left [ \frac{\left| \sum_{x \in \{0,1\}^\ell} (2\bff(x) - 1) \right|}{2^\ell} \right] - \E_{\bgg \sim \calD_2} \left[ \frac{\left| \sum_{x \in \{0,1\}^\ell} (2\bgg(x) - 1)\right|}{2^\ell} \right]  \right| \leq 2 \eps + 8 e^{- \Omega(\eps^2 m)}. \]
	Finally, note that for a function $f:\{0,1\}^\ell \rightarrow \{0,1\}$
		\[\dist(f, \{0,1\}) = \frac{1}{2} - \frac{1}{2} \cdot \frac{\left| \sum_{x \in \{0,1\}^\ell} (2f(x) - 1) \right|}{2^\ell}. \]
	Thus
		\[\bigg| \E_{\bff \sim \calD_1} \left [\dist(\bff, \{0,1\}) \right] - \E_{\bgg \sim \calD_2} \left[\ \dist(\bgg, \{0,1\}) \right]  \bigg| \leq \eps + 4e^{- \Omega(\eps^2 m)}. \]
Taking $\eps = \Theta \left( \sqrt{\frac{\log(m)}{m}} \right)$ then gives the result.
\end{proof}

\begin{proof}[Proof of Lemma \ref{lem:testing-random-functions}]
We consider the following algorithm, $\ALG$: Set $r = \log(8/\eps^2) + \log\log(32/\eps)$. We'll assume that $\eps \geq 1000 \cdot 2^{-(n-k)/10}$, which implies $r \leq n-k$. Sample $m := 1024n^2/\eps^2$ random points $\bx^{(1)}, ..., \bx^{(m)}$ and query every point in a ball of radius $r$ around each $\bx^{(i)}$. For each set $S \subseteq [n]$ of size $k$, compute 
	\[e_S(\bx^{(i)}) := \frac{1}{|B_{n-k}(r)|} \left| \sum_{z \in B(\bx^{(i)},r) \cap \{y \in \{0,1\}^n: y_S = \bx_S^{(i)} \} } (2\bff(z) - 1) \right| \]
	If for some set $S$, 
	\[ \frac{1}{m} \sum_{i = 1}^m \mathbb{I}(e_S(\bx^{(i)}) > \eps/2) \geq \eps^2/8\]
	then accept. Otherwise, reject.
	
	\begin{claim}
		\[\Pr_{\bff \sim \calD_{YES}}[\ALG(\bff) \text{ accepts}] = 1 - o(1) \]	
	\end{claim}
	
	\begin{proof}
		Let $f$ be a function that is $(\frac{1}{2} - \eps)$-close to some junta on a set of relevant variables $S$ of size $k$. 
		Note
			\[\dist(f, \calJ_S) = \frac{1}{2^{k}}\sum_{\rho \in \{0,1\}^S} \frac{1}{2} \left(1 - \frac{1}{2^{n-k}} \left| \sum_{y \in \{0,1\}^{[n] \setminus S}}  2f(\rho \sqcup y) - 1 \right| \right) \leq \frac{1}{2} - \eps. \]
		Rearranging, we have that
			\[ \frac{1}{2^{k}}\sum_{\rho \in \{0,1\}^S} \frac{1}{2^{n-k}} \left| \sum_{y \in \{0,1\}^{[n] \setminus S}}  2f(\rho \sqcup y) - 1 \right| \geq 2\eps. \]
		By an averaging argument, it now follows that for at least $\eps 2^k$ values of $\rho$ we have that
			\[ \frac{1}{2^{n-k}} \left| \sum_{y \in \{0,1\}^{[n] \setminus S}}  2f(\rho \sqcup y) - 1 \right| \geq \eps. \]
			
		Fix a particular $\rho \in \{0,1\}^S$ such that the above holds and assume that 
			\[ \frac{1}{2^{n-k}}  \sum_{y \in \{0,1\}^{[n] \setminus S}}  (2f(\rho \sqcup y) - 1) \geq \eps. \]
		Now observe
			\[\frac{1}{2^{n-k}} \sum_{y \in \{0,1\}^{[n] \setminus S}} \frac{1}{|B_{n-k}(r)|} \sum_{x \in B(y,r)} (2f(\rho \sqcup x)-1) = \frac{1}{2^{n-k}} \sum_{y \in \{0,1\}^{[n] \setminus S}} (2f(\rho \sqcup y)-1) \geq \eps. \]
		So by another averaging argument, we get that for at least $\eps 2^{n-k-1}$ values of $y \in \{0,1\}^{[n] \setminus S}$
			\[	\frac{1}{|B_{n-k}(r)|} \left| \sum_{x \in B(y,r)} 2f(\rho \sqcup x) -1 \right| \geq \eps/2 \]
		An analogous argument gives the same result when	
			\[ \frac{1}{2^{n-k}}  \sum_{y \in \{0,1\}^{[n] \setminus S}}  (2f(\rho \sqcup y) - 1) \leq -\eps. \]
				
		We can now conclude that
			\[\Pr_{\bx_i}[e_{S}(\bx^{(i)}) \geq \eps/2] \geq \eps^2/2.\]
		So by a Chernoff bound we conclude that	
			\[\Pr_{\bx^{(1)}, ..., \bx^{(m)}} \left [\frac{1}{m} \sum_{i = 1}^m \mathbb{I}(e_S(\bx^{(i)}) > \eps/2) \leq \eps^2/8 \right] \leq e^{-m\eps^2/8}.\]
		Since $m \geq 1024 n^2/\eps^2$, we conclude that we accept such a function $f$ with high probability. Finally, as a function $\bff \sim \calD_{YES}$ is $(\frac{1}{2} - \eps)$ close to a $k$-junta with high probability the result follows.
\end{proof}
	
	\begin{claim}
		\[\Pr_{\bff \sim \calD_{NO}}[\ALG(\bff) \text{ accepts}] = o(1)\]	
	\end{claim}
	
	\begin{proof}
		Fix a set $S \subseteq [n]$ of size $k$. By a Hoeffding bound, we have that 
			\[ \Pr_{\bff \sim \calD_{NO}, \bx^{(i)}} \left [ e_S(\bx^{(i)}) \geq \eps/2 \right] \leq 2e^{-\eps^2 |B_{n-k}(r)|/8} \leq 2e^{-\eps^2 2^r/8} \leq 2e^{-\log(32/\eps)} = \eps/16 \]
				
		Since $r \leq n/4$, conditioned on $\dist(\bx^{(i)}, \bx^{(j)}) \geq n/4$ for all $i \not = j$, we have that the events $\mathbb{I}(e_S(\bx^{(i)}) \geq \eps/2)$ are independent. (Note that without this assumption the events $\mathbb{I}(e_S(\bx^{(i)}) \geq \eps/2)$ are only independent given $f$.) Applying Chernoff bounds then gives us that
				\[\Pr_{\bff \sim \calD_{NO}} \left[\frac{1}{m} \sum_{i = 1}^m \mathbb{I}(e_S(\bx_i) > \eps/2) \geq \eps^2/8 \bigg | \dist(\bx^{(i)}, \bx^{(j)}) \geq n/4 \quad \forall i \not = j \right] \leq e^{-\eps m/64}.\]
		Taking a union bound over all sets $S$,
			\[\Pr_{\bff \sim \calD_{NO}} \left[\ALG(\bff) \text{ accepts} \big | \dist(\bx^{(i)}, \bx^{(j)}) \geq n/4 \quad \forall i \not = j \right] \leq n^k e^{-\eps^2 m/64} = o(1).\]
		It now remains to compute the probability that the $\bx^{(i)}$'s are far from one another. By a Chernoff bound
			\[\Pr_{\bx^{(1)}, \bx^{(2)}} \left [ \dist(\bx^{(1)}, \bx^{(2)}) \leq n/4 \right] \leq e^{-n/16}. \]
		A union bound then yields
			\[\Pr_{\bx^{(1)}, \bx^{(2)}, ..., \bx^{(m)}} \left [\exists i \not = j: \quad  \dist(\bx^{(i)}, \bx^{(j)}) \leq n/4 \right] \leq m^2 e^{-n/16} \leq 2^{20} \cdot n^4 \cdot 2^{n/32} \cdot e^{-n/16} = o(1). \]

	\end{proof}
\end{proof}